\documentclass[journal]{IEEEtran}
\ifCLASSINFOpdf
\else
\fi
\usepackage{amsmath,amsthm,amssymb,bm,enumerate,soul,cancel,graphicx,booktabs,float,hyperref,multirow,xcolor,subcaption}
\usepackage[scr=boondoxo,scrscaled=1.05]{mathalfa}
\usepackage{mdframed}
\usepackage[linesnumbered,ruled,vlined]{algorithm2e}

\newtheorem{theorem}{Theorem}
\newtheorem{definition}{Definition}
\newtheorem{corollary}{Corollary}
\newtheorem{condition}{Condition}
\newtheorem{lemma}{Lemma}

\begin{document}
%
\title{Sufficient Conditions for Parameter Convergence over Embedded Manifolds using Kernel Techniques}
%
%
%

\author{Sai~Tej~Paruchuri,~
        Jia~Guo,~
        and~Andrew~Kurdila
\thanks{M. Shell was with the Department
of Electrical and Computer Engineering, Georgia Institute of Technology, Atlanta,
GA, 30332 USA e-mail: (see http://www.michaelshell.org/contact.html).}
\thanks{J. Doe and J. Doe are with Anonymous University.}
\thanks{Manuscript received April 19, 2005; revised August 26, 2015.}}

%
%

\markboth{Journal of \LaTeX\ Class Files,~Vol.~14, No.~8, August~2015}%
{Shell \MakeLowercase{\textit{et al.}}: Bare Demo of IEEEtran.cls for IEEE Journals}
%



\maketitle

\begin{abstract}
The persistence of excitation (PE) condition is sufficient to ensure parameter convergence in adaptive estimation problems. Recent results on adaptive estimation in reproducing kernel Hilbert spaces (RKHS) introduce PE conditions for RKHS. This paper presents sufficient conditions for PE for the particular class of uniformly embedded reproducing kernel Hilbert spaces (RKHS) defined over smooth Riemannian manifolds. This paper also studies the implications of the sufficient condition in the case when the RKHS is finite or infinite-dimensional. When the RKHS is finite-dimensional, the sufficient condition implies parameter convergence as in the conventional analysis. On the other hand, when the RKHS is infinite-dimensional, the same condition implies that the function estimate error is ultimately bounded by a constant that depends on the approximation error in the infinite-dimensional RKHS. We illustrate the effectiveness of the sufficient condition in a practical example.
\end{abstract}


%
\IEEEpeerreviewmaketitle

\section{Introduction}
\label{sec_intro}
Adaptive estimation of unknown nonlinearities arising in finite-dimensional autonomous dynamical systems is now a classical, or textbook, problem. \cite{Sastry2011, Narendra2012, Ioannou} Typically, in such estimation problems, we assume that the unknown function is a linear combination of known basis functions, commonly referred to as regressors. We can guarantee parameter convergence in adaptive estimation problems by assuming that additional sufficient conditions on the regressors hold. In control theory parlance, we refer to these hypotheses as persistence of excitation (PE) conditions. It is important to observe that the definition of PE can vary depending on the type of problem or the algorithm implemented.

\subsection{PE Conditions and Convergence}
\label{ssec_PEandConv}
Some of the earliest accounts of PE conditions and their implications for parameter convergence in adaptive estimation are given in \cite{Morgan1977, Morgan1977b}. In these studies, the authors pose the problem of parameter estimation as the stability analysis of a linear time-varying (LTV), finite-dimensional system and show that the LTV systems are asymptotically stable when the PE condition holds. In some cases, we can even guarantee exponential stability. \cite{Boyd1983, Anderson1977} When only a subspace is persistently excited, it is possible to show that the parameter error eventually becomes orthogonal to that subspace. \cite{Boyd1983} In other words, the estimates converge to the projection of the unknown function onto the subspace.

In \cite{Panteley2001, Loria2003, Panteley2000a, Loria2003a, Panteley1998}, the authors generalize some of the existing notions of PE and extend the theory on the stability of LTV systems. The work by Farrell illustrates the effectiveness of local PE conditions for parameter convergence. \cite{Farrell1997} The PE condition in \cite{Novara2011} ensures the convergence of parameter estimates of systems defined by the interconnection of LTI blocks and nonlinear functions. Yuan and Wang relate the learning speeds and constants that appear in the PE definition in \cite{Yuan2011}. In \cite{Nikitin2007}, Nikitin proposes a generalized PE definition that relaxes some of the conditions imposed on conventional PE conditions. An account of parameter estimation for distributed parameter systems and the corresponding notion of PE is given in \cite{Demetriou1994, Demetriou1994a, Demetriou2006, Baumeister1997, Bohm1998, Demetriou1994b}. Recent articles on adaptive estimation in reproducing kernel Hilbert spaces (RKHS) extend the notion of partial PE to cases where the unknown function appearing in ordinary differential equations (ODEs) belongs to an infinite-dimensional RKHS. \cite{Bobade2019, jia2020a, jia2020b}

The PE condition is difficult, and sometimes impossible, to verify \emph{a priori} in practical applications. To overcome this limitation, authors have studied simpler sufficient conditions that ensure PE. A few of the earliest accounts that analyze sufficient conditions for PE are \cite{Bai1985, Boyd1986}. In these papers, the authors link the richness of the reference trajectory to the PE condition. This richness condition is much more intuitive than the PE condition. Kurdila et al. illustrate in \cite{Kurdila1994, Kurdila1995} that, in function spaces generated by radial basis functions, radial basis functions centered at points in state space that are \emph{visited regularly} are persistently excited. The work by Gorinevsky and Lu et al., in which it is shown that inputs belonging to neighborhoods of radial basis function centers are PE, illustrates a similar result. \cite{Gorinevsky1995, Lu1998} In \cite{Wang2006, Wang2009}, Wang et al. relax some of the hypotheses in \cite{Kurdila1994, Kurdila1995}. They derive a sufficient condition for PE for any recurrent trajectory contained in a regular lattice. In \cite{Bamieh2002}, Bamieh and Giarre pose a linear parameter varying identification problem as linear regression, which allows them to show that in some instances, the PE condition simplifies to an interpolation condition.

An alternative to developing sufficient conditions that ensure PE is to develop methods that ensure parameter convergence without PE. For example, Adetola and Guay prove in \cite{Adetola2008} that we can compute the unknown parameters once the regressor matrix becomes positive definite. The recent class of estimation techniques, referred to as concurrent learning, obviates the need for persistently exciting signals by using a rich collection of recorded data. \cite{Chowdhary2014, Vamvoudakis2016, Kersting2019} Song et al. show asymptotic constancy of parameter estimates without the PE condition in \cite{Song2017}. In \cite{Wang2020}, Wang et al. propose a finite-time parameter estimation technique that uses the dynamic regressor extension and mixing methods to transform the estimation problem into a series of regression models. This transformation results in parameter convergence under non-PE conditions.

In most of the studies mentioned above, the states and the parameters evolve in Euclidean spaces (a notable exception being the family of related papers \cite{Demetriou1994, Demetriou1994a, Demetriou2006, Baumeister1997, Bohm1998, Demetriou1994b}, which treat distributed parameter systems). However, for a given initial condition in many such finite-dimensional systems, the state trajectory traverses only a subdomain of Euclidean space. Recent articles on adaptive estimation in RKHS provide a framework for adaptive estimation of dynamic systems whose states evolve in more generic spaces, including embedded manifolds. \cite{Bobade2019, Kurdila2019PE} The corresponding PE conditions are given in \cite{jia2020a, jia2020b}. Guo et al. study the rate of convergence of the finite-dimensional approximations of reproducing kernel Hilbert spaces defined over manifolds in \cite{Guo2020Rates}. From an adaptive estimation perspective, this is equivalent to studying the rate at which the finite-dimensional function estimate $\hat{f}_n$ converges to the infinite-dimensional function estimate $\hat{f}$. However, to carryout the analysis in these studies, it is necessary to choose the RKHS that is persistently excited.

This requirement, in turn, suggests a need for sufficient conditions for PE that works in spaces that are more general than Euclidean spaces. In this paper, we introduce a sufficient condition for PE of RKHS defined over embedded manifolds and study its implications in both finite and infinite-dimensional cases. In the typical situation in which the analysis in this paper is applied, we assume that we are given an ODE (such as in the model problem Equation \ref{eq_plant}), and that the system admits an invariant submanifold $M$ that is regularly embedded in the state space $\mathbb{R}^d$ for some given initial condition. The sufficient condition is also applicable when, given an initial condition, the forward orbit is a subset of an invariant manifold and/or the embedded manifold is Euclidean space itself.

\subsection{Summary of New Results}
\label{ssec_newresults}
This paper extends the results in the recent papers \cite{Bobade2019, Kurdila2019PE, jia2020a, jia2020b, Guo2020Rates} in several fundamental ways. The first result states sufficient conditions that guarantee the PE condition for finite-dimensional RKHS over a manifold $M$ that is defined in terms of a finite number of kernel basis functions. While this was carried out in \cite{Kurdila1995} for radial basis functions over $\mathbb{R}^n$, here we treat the case where the native space is defined over a smooth manifold and the RKHS is generated by a continuous strictly positive definite kernel. As in \cite{Kurdila1995}, we see that the RKHS is PE if the trajectory repeatedly visits any (geodesic) neighborhoods of the kernel basis centers, and the time of visitation is bounded below in some sense. This result has direct applicability to finite-dimensional cases of the RKHS embedding methods discussed in \cite{Bobade2019, Kurdila2019PE, jia2020a, jia2020b, Guo2020Rates}. It serves as a foundation for practical choices of PE subsets and spaces, which is not addressed in these references. The second principal result of this paper is the study of the implications of the above sufficient condition when the RKHS is infinite-dimensional. We show that when the sufficient condition described above is valid, the function (parameter) error is eventually bounded above by a constant, which depends on the finite-dimensional approximation error of the infinite-dimensional RKHS. Researchers have investigated such cases for parameter estimation in Euclidean spaces using dead-zone gradient algorithm. \cite{Sanner1992b,Ioannou} The result in this paper can be considered as a generalization of this approach to reproducing kernel Hilbert spaces of functions defined over manifolds.

The organization of this paper is as follows. Section \ref{sec_adapEstRKHS} reviews background material for the new results in Sections \ref{sec_suffPE} and \ref{sec_suffCondImp}. It covers the theory of adaptive estimation in reproducing kernel Hilbert spaces and introduces two recent, different notions of the persistence of excitation. We discuss when the two notions of PE are equivalent and when we can ensure parameter convergence. In Section \ref{sec_suffPE}, we derive one of the primary results of this paper, a sufficient condition for PE in the finite-dimensional case. We show that this sufficient condition ensures convergence of parameters when the unknown function belongs to a known finite-dimensional space. In Section \ref{sec_suffCondImp}, we discuss the implications of this sufficient condition when we only know that the unknown function belongs to an infinite-dimensional space. We show that the projection (onto the persistently excited finite-dimensional subspace) of the function estimate error is bounded by a constant times the error of best approximation. Section \ref{sec_numIllus} illustrates the theory using a numerical example. Section \ref{sec_conc} concludes the paper.

%
\section{Review of Adaptive Estimation in RKHS}
\label{sec_adapEstRKHS}
\subsection{The Theory of RKHS}
A reproducing kernel Hilbert space $\mathcal{H}_X$ is a Hilbert space of functions defined on the set $X$ and that can be defined in terms of an associated continuous, positive-definite kernel $\mathcal{K}: X \times X \to \mathbb{R}$. In this paper, we assume that the kernel is strictly positive-definite. For each $\bm{x} \in X$, the kernel basis centered at $\bm{x}$, denoted $\mathcal{K}(\bm{x},\cdot)$, is a function in $\mathcal{H}_X$. Suppose $\left( \cdot, \cdot \right)_{\mathcal{H}_X}$ is the inner product associated with the space $\mathcal{H}_X$. The reproducing property of the RKHS states that for any $\bm{x} \in X$ and $f \in \mathcal{H}_X$, $(\mathcal{K}(\bm{x},\cdot), f)_{\mathcal{H}_X} = \mathcal{E}_{\bm{x}} f = f(\bm{x})$. The operator $\mathcal{E}_{\bm{x}}$ is the evaluation functional. In the context of this paper, we assume that the evaluation operator is \emph{uniformly bounded}, that is $|\mathcal{E}_{\bm{x}}f| \leq c \|f\|_{\mathcal{H}_X}$ for all $\bm{x} \in X$ and $f \in \mathcal{H}_X$ and some fixed positive constant $c$. One sufficient condition for the uniform boundedness of all the evaluation functionals is that there exists a constant $\bar{k}$ such that $\mathcal{K}(\bm{x},\bm{x}) \leq \bar{k}^2 < \infty$ for all $\bm{x} \in X$. This condition implies that the RKHS is continuously embedded in the space of continuous functions $C(X)$ defined on $X$, that is, given any function $f \in \mathcal{H}_X$, we have $\|f\|_{C(X)} \leq c \|f\|_{\mathcal{H}_X}$ for some constant $c$. Given a positive definite kernel $\mathcal{K} : X \times X \to \mathbb{R}$, we generate the associated RKHS by
\begin{equation*}
\mathcal{H}_X := \overline{ span \{ \mathcal{K}(\bm{x},\cdot) | \bm{x} \in X \}},
\end{equation*}
where the inner product satisfies $\langle\mathcal{K}(\bm{x},\cdot), \mathcal{K}(\bm{y},\cdot) \rangle_{\mathcal{H}_X} = \mathcal{K}(\bm{x},\bm{y})$. For any set $M \subseteq X$, the associated RKHS $\mathcal{H}_\Omega \subseteq \mathcal{H}_X$ is defined as
\begin{equation*}
\mathcal{H}_\Omega := \overline{ span \{ \mathcal{K}(\bm{x},\cdot) | \bm{x} \in M \subseteq X \}}.
\end{equation*}
Additionally, if $\Omega_n$ is a discrete finite set of $n$ elements in $X$, the associated RKHS is an $n$-dimensional space. In this paper, we use a subscript, as in $\Omega_n$, to describe the number of elements in a discrete finite set. In addition to the above spaces, we are also interested in the space of restrictions $R_M (\mathcal{H}_X)$, where $M \subseteq X$. We define the space $R_M (\mathcal{H}_X)$ by
\begin{align*}
    R_M(\mathcal{H}_X):=\{g:M \rightarrow \mathbb{R} \ | \ g=R_M f:=f|_M ~\forall~ f\in \mathcal{H}_X \}. 
\end{align*} 
The functions in $R_M (\mathcal{H}_X)$ are defined only on $M \subseteq X$, but those in $\mathcal{H}_X$ are defined everywhere in $X$. When $M = X$, the space $R_M (\mathcal{H}_X)$ is nothing but the space $\mathcal{H}_X$. The restricted space $R_M (\mathcal{H}_X)$ is itself an RKHS, \cite{Berlinet2011,Saitoh2016} and the associated reproducing kernel is given by 
\begin{align*}
    \mathcal{R}(\bm{x},\bm{y})=\mathcal{K}|_M(\bm{x},\bm{y})=\mathcal{K}(\bm{x},\bm{y})
\end{align*}
for all $\bm{x},\bm{y}\in M$. The kernel $\mathcal{R}$ generates the space $R_M (\mathcal{H}_X)$ just as the kernel $\mathcal{K}$ generates $\mathcal{H}_X$.

In this paper, the set $X$ represents the state space $\mathbb{R}^d$ of the plant, and the set $M$ is taken to be a smooth, Riemmanian, $k$-dimensional manifold that is regularly embedded in the state space $X$. The sets $\Omega$ and $\Omega_n$ are used to represent persistently excited subsets of $X$. The reproducing property mentioned above endows the RKHS $\mathcal{H}_X$ with many interesting properties which makes proving theorems easier. A detailed description of these properties is given in \cite{Aronszajn1950, Berlinet2011,Saitoh2016}. We describe the properties as and when we use them in this paper.

%
\subsection{RKHS Embedding for Adaptive Estimation}
In this subsection and the next, we discuss several recent results that are critical to the new results derived in Sections \ref{sec_suffPE} and \ref{sec_suffCondImp}. Interested readers are referred to \cite{Bobade2019, jia2020a, jia2020b} for more detailed discussions. Suppose we have a nonlinear system governed by the ordinary differential equation
\begin{align}
\label{eq_plant}
\dot{\bm{x}}(t) = A \bm{x}(t) + B f(\bm{x}(t)),
\end{align}
where $\bm{x}(t) \in X := \mathbb{R}^d$, $A \in \mathbb{R}^{d \times d}$ is known and Hurwitz, $B \in \mathbb{R}^d$ is known and $f : \mathbb{R}^d \to \mathbb{R}$ is the unknown (nonlinear) function, which is assumed to be an element of the RKHS $\mathcal{H}_X$. We also assume that we measure all the states $\bm{x}(t)$ of the system at each time $t \geq 0$. We define an estimator model of the form
\begin{align}
\label{eq_plntEst}
    \dot{\hat{\bm{x}}}(t) = A \hat{\bm{x}}(t) + B \hat{f}(t,\bm{x}(t)),
\end{align}
where $\hat{\bm{x}}(t) \in \mathbb{R}^d$ is the state estimate and $\hat{f}(t,\bm{x}(t))$ is the function estimate, which at each time $t$ is an element of the RKHS $\mathcal{H}_X$. Our goal is to ensure that the function estimate $\hat{f}(t)$ approaches the true function $f$ as $t \to \infty$. We use the gradient learning law, which is given by
\begin{align}
    \dot{\hat{f}}(t) = \Gamma^{-1} (B \mathcal{E}_{\bm{x}(t)})^* P (\bm{x}(t) - \hat{\bm{x}}(t)),
    \label{eq_lLaw}
\end{align}
to define the rate of change of the function estimate. In the above equation, the term $\Gamma \in \mathbb{R}$ and the notation $L^*$ represents the adjoint of the linear operator $L$. The matrix $P$ is the symmetric positive definite solution of the Lyapunov's equation $A^T P + PA = - Q$, where $Q \in \mathbb{R}^{d \times d}$ is an arbitrary symmetric positive-definite matrix. 

It is now possible to write down the error equations, which have the form
\begin{align}
    \begin{Bmatrix}
    \dot{\tilde{\bm{x}}}(t) \\ \dot{\tilde{f}}(t)
    \end{Bmatrix}
    = 
    \begin{bmatrix}
    A & B \mathcal{E}_{\bm{x}(t)} \\
    - \Gamma^{-1} (B \mathcal{E}_{\bm{x}(t)})^* P & 0
    \end{bmatrix}
    \begin{Bmatrix}
    \tilde{\bm{x}}(t) \\ \tilde{f}(t)
    \end{Bmatrix}.
    \label{eq_errEst}
\end{align}
In the above equation, the terms $\tilde{\bm{x}}(t) := \bm{x}(t) - \hat{\bm{x}}(t)$ and $\tilde{f}(t,\cdot) := f(\cdot) - \hat{f}(t,\cdot)$ represent the state and function error, respectively. The error systems, governed by the above equations, evolves in the infinite-dimensional space $\mathbb{R}^d \times \mathcal{H}_X$. Lyapunov analysis and Barbalat's lemma can be used to show that the state error $\tilde{\bm{x}}(t)$ converges to zero. \cite{jia2020a, jia2020b} However, we cannot make any claims about the function error $\tilde{f}$ without additional assumptions.

%
\subsection{Persistence of Excitation}
\label{ssec_PE}
As in the study of finite dimensional systems in \cite{Sastry2011, Narendra2012, Ioannou}, persistence of excitation conditions introduced in \cite{jia2020a, jia2020b} for RKHS embedding are sufficient to prove convergence of the function error $\tilde{f}(t) \to 0$. We discuss two notions of PE conditions.

\begin{definition}
\label{def_PE1}
(\textbf{PE $\mathcal{H}_X$-$1$}) The trajectory $\bm{x}: t \mapsto \bm{x}(t) \in \mathbb{R}^d$ persistently excites the indexing set $\Omega$ and the RKHS $\mathcal{H}_\Omega$ provided there exist positive constants $T_1, \gamma_1, \delta_1,$ and $\Delta_1$, such that for each $t \geq T_1$ and any $g \in \mathcal{H}_X$, there exists $s \in [t,t+\Delta_1]$ such that
\begin{align*}
    \left| \int_{s}^{s+\delta_1} \mathcal{E}_{\bm{x}(\tau)} g d \tau \right| \geq \gamma_1 \| P_\Omega g \|_{\mathcal{H}_X} > 0.
\end{align*}
\end{definition}

\begin{definition}
\label{def_PE2}
(\textbf{PE $\mathcal{H}_X$-$2$}) The trajectory $\bm{x}: t \mapsto \bm{x}(t) \in \mathbb{R}^d$ persistently excites the indexing set $\Omega$ and the RKHS $\mathcal{H}_\Omega$ provided there exists positive constants $T_2$, $\gamma_2$, and $\Delta_2$ such that 
\begin{align*}
    \int_t^{t+\Delta_2} \left\langle \mathcal{E}^*_{\bm{x}(\tau)} \mathcal{E}_{\bm{x}(\tau)}g,g \right\rangle_{\mathcal{H}_X} d \tau \geq \gamma_2 \| P_\Omega g \|_{\mathcal{H}_X}^2 > 0
\end{align*}
for all $t \geq T_2$ and any $g \in \mathcal{H}_X$.
\end{definition}

The space $\mathcal{H}_X$ in the notation ``PE $\mathcal{H}_X$-$1$" and ``PE $\mathcal{H}_X$-$2$" refers to the space in which the functions $g$ are contained. The operator $P_\Omega$ is the $\mathcal{H}_X$-orthogonal projection from the space $\mathcal{H}_X$ onto the closed subspace $\mathcal{H}_\Omega$. The following theorem shows that the function error converges over the PE set when the PE condition in Definition \ref{def_PE1} holds.

\begin{theorem}
\label{thm_paramConv}
If the trajectory $\bm{x} : t \mapsto \bm{x}(t)$ persistently excites the RKHS $\mathcal{H}_\Omega$ in the sense of Definition PE $\mathcal{H}_X -$ \ref{def_PE1}. Then 
\begin{align*}
    \lim_{t \to \infty} \| \tilde{\bm{x}}(t) \| = 0, \hspace{0.75in} \lim_{t \to \infty} \| P_\Omega \tilde{f}(t) \|_{\mathcal{H}_X} = 0.
\end{align*}
\end{theorem}
In the above theorem, we can additionally show that if $\lim_{t \to \infty} \| P_\Omega \tilde{f}(t) \|_{\mathcal{H}_X} = 0$, then $\lim_{t \to \infty} | f(\bm{x}) - \hat{f}(t,\bm{x}) | = 0$ for all $\bm{x} \in \Omega$. In fact, the convergence is uniform over the set $\Omega$ since we assume that the evaluation functional is uniformly bounded.

Before proceeding further, let us note how the above definitions and theorem simplify when the actual unknown function $f \in \mathcal{H}_{\Omega_n}$. In such cases, we can assume that the function $\hat{f}$ in the adaptive estimator equation and the functions $g$ in Definitions \ref{def_PE1} and \ref{def_PE2} are in the space $\mathcal{H}_{\Omega_n}$ and revise the definitions of PE conditions to PE $\mathcal{H}_{\Omega_n}$-$1$ and PE $\mathcal{H}_{\Omega_n}$-$2$. Since the trajectory $\bm{x}: t \mapsto \bm{x}(t) \in \mathbb{R}^d$ persistently excites the space $\mathcal{H}_{\Omega_n}$ and all the functions are in $\mathcal{H}_{\Omega_n}$, the error equations can be recast in $\mathbb{R}^d \times \mathcal{H}_{\Omega_n}$, and the projection operator $P_\Omega \equiv P_{\Omega_n}$ disappears. On the other hand, when the evolution of the state trajectory is on a manifold $M$, we can treat the above problem solely as estimation of functions over the manifold $M$. In such cases, the persistently excited set $\Omega$ is a subset of the manifold and we replace the space $\mathcal{H}_X$ and $\mathcal{H}_{\Omega_n}$ with $R_M(\mathcal{H}_X)$ and $R_M(\mathcal{H}_{\Omega_n})$, respectively, in the above theorems and definitions.

%
\subsection{Equivalence of PE conditions}
\label{ssec_PEequiv}
In the previous subsection, we discussed two different notions of PE. PE $\mathcal{H}_X$-$1$ always implies PE $\mathcal{H}_X$-$2$. 
\begin{theorem}
The PE condition in Definition \ref{def_PE1} implies the one in Definition \ref{def_PE2}.
\end{theorem}
The proof of  the above theorem is given in \cite{jia2020a,jia2020b}. Now, note that the hypotheses of Theorem \ref{thm_paramConv} assumes that the PE condition in Definition \ref{def_PE1} holds. On the other hand, the sufficient condition, given in next section, implies that the PE condition in Definition \ref{def_PE2} holds. In particular, it implies that PE $\mathcal{H}_{\Omega_n}$ - $\bm{1}$ holds, where $\mathcal{H}_{\Omega_n}$ is a finite-dimensional RKHS. Thus, it is important to understand when the PE condition in Definition \ref{def_PE2} implies the PE condition in Definition \ref{def_PE1}. The following theorem from \cite{jia2020a, jia2020b} explicitly states when the two notions of PE are equal.

\begin{theorem}
If the family of functions defined by $\mathbb{U}(\bar{S}_n) = \{ g(\bm{x}(\cdot)): t \mapsto g(\bm{x}(t)) | g \in \bar{S}_n := \mathcal{H}_{\Omega_n} \text{ such that } \| g \|_{\mathcal{H}_{\Omega_n}} = 1 \}$ is uniformly equicontinuous, then the PE $\mathcal{H}_{\Omega_n}$ - \ref{def_PE2} implies PE $\mathcal{H}_{\Omega_n}$ - \ref{def_PE1}.
\end{theorem}
The proof of a more general case of the above theorem is given in \cite{jia2020a}. It is important to understand the family of functions $\mathbb{U}(\bar{S}_n)$ are equicontinuous. A sufficient condition for this is that the unit ball $\bar{S}_n = \{g: X \mapsto \mathbb{R} \in \mathcal{H}_{\Omega_n} \text{ such that } \|g\| = 1 \}$ is uniformly equicontinuous and the state trajectory $t \mapsto \bm{x}(t)$ is uniformly continuous. If the state trajectory $t \mapsto \bm{x}(t)$ maps to a compact set $V$, then $\mathbb{U}(\bar{S}_n)$ is uniformly equicontinuous if $\bar{S}_{n}$ redefined as $\bar{S}_{n} = \{g: V \mapsto \mathbb{R} \in \mathcal{H}_{\Omega_n} \text{ such that } \|g\| = 1 \}$ is uniformly equicontinuous and the state trajectory $t \mapsto \bm{x}(t)$ is uniformly continuous. We know that $\bar{S}_{V,n}$ is uniformly equicontinuous. Thus, if the state trajectory $t \mapsto \bm{x}(t)$ is uniformly continuous and maps to a compact set, the family of functions $\mathbb{U}(\bar{S}_n)$ is uniformly equicontinuous.

%
\section{Sufficient Condition for PE}
\label{sec_suffPE}
In this section, we derive the sufficient condition for persistence of excitation of the trajectory $\bm{x} : t \mapsto \bm{x}(t)$ in the sense of the PE $\mathcal{H}_{\Omega_n}$ - \ref{def_PE2}. We assume that the states evolves in a \emph{smooth, compact, Riemmanian $k$-dimensional} manifold that is regularly embedded in $X$ and endowed with the (Riemmanian) distance function $d_M(\cdot,\cdot) : M \times M \to \mathbb{R}^+ \cup \{ 0 \}$. Note, by definition, $d_M(\bm{x},\bm{y})$ is equal to the infimum of the lengths of all the smooth curves joining $\bm{x} \in M$ and $\bm{y} \in M$. The sufficient condition is valid for the case when $f \in R_M(\mathcal{H}_{\Omega_n}) \subseteq R_M(\mathcal{H}_X)$, where $\Omega_n = \{ \bm{x}_1,\ldots,\bm{x}_n \}$ is a discrete finite set in $M$. We analyze the implications of relaxing this condition in the next section. In the following analysis, we assume that the kernel $\mathcal{R}:M \times M \to \mathbb{R}$ is a continuous, \emph{strictly} positive-definite kernel. Many kernels are strictly positive definite (Matern/Sobolev, exponential, multiquadric).



\begin{lemma}
\label{lem_knlMatBnd}
Suppose $\bm{y}_i \in M$ for $i = 1,\ldots,n$.
If
\begin{align}
\label{eq_knlMatDef}
    S(\bm{y}_1,\ldots,\bm{y}_n) :=
    \begin{pmatrix}
    \mathcal{R}(\bm{x}_1,\bm{y}_1)& \dots & \mathcal{R}(\bm{x}_n,\bm{y}_1) \\
    \vdots & \ddots & \vdots \\
    \mathcal{R}(\bm{x}_1,\bm{y}_n)& \dots & \mathcal{R}(\bm{x}_n,\bm{y}_n)
    \end{pmatrix},
\end{align} 
then there exists an $\epsilon > 0$ and a number $\theta(\epsilon,\bm{x}_1,\ldots,\bm{x}_n) > 0$ such that 
\begin{align*}
    \| S \bm{\alpha} \| \geq \theta \| \bm{\alpha} \|
\end{align*}
for all $\bm{\alpha} \in \mathbb{R}^n$ and for every collection of $\bm{y}_i$'s that satisfy $d_M(\bm{x}_i,\bm{y}_i) \leq \epsilon$ for $i = 1,\ldots,n$.
\end{lemma}

\begin{proof}
The proof of this lemma follows easily by modifications of the arguments in \cite{Kurdila1995} (which holds for radial basis functions in $\mathbb{R}^n$) to the case when the basis function is a continuous, \emph{strictly} positive-definite kernel basis function defined on a manifold. We note that the eigenvalues of the matrix $S^T S$ vary continuously with $\bm{y}_i$ for $i = 1,\ldots,n$, since the eigenvalues are continuous functions of the elements of a matrix and the map $\bm{y} \to \mathcal{R}(\bm{x},\bm{y})$ is continuous by hypothesis. Let $\lambda(\bm{y}_1,\ldots,\bm{y}_n)$ be the smallest eigenvalue of $S(\bm{y}_1,\ldots,\bm{y}_n)^T S(\bm{y}_1,\ldots,\bm{y}_n)$. Since the kernel is strictly positive definite, the smallest eigenvalue of $S(\bm{x}_1,\ldots,\bm{x}_n)^T S(\bm{x}_1,\ldots,\bm{x}_n)$ satisfies $\lambda(\bm{x}_1,\ldots,\bm{x}_n) > 0$. By continuity of eigenvalues, we choose an $\epsilon > 0$ such that 
\begin{align*}
    \lambda(\bm{y}_1,\ldots,\bm{y}_n) > \frac{1}{2} \lambda(\bm{x}_1,\ldots,\bm{x}_n) > 0
\end{align*}
whenever $d_M(\bm{x}_i,\bm{y}_i) \leq \epsilon$ for $i = 1\ldots,n$. (It is easy to see that such a choice is always possible. Since $y \mapsto \lambda(y)$ is continuous at $x$, for any $\gamma > 0$, there is an $\epsilon > 0$ such that if $d_M(x,y) < \epsilon$, then $|\lambda(y) - \lambda(x)| < \gamma$. Choose $\gamma:= \frac{1}{2} \lambda(x)$, and pick an appropriate $\epsilon > 0$. Then the smallest that $\lambda(y)$ can be is greater than $\frac{1}{2} \lambda(x)$. So, $ \lambda(y) \geq \frac{1}{2} \lambda(x) > 0$.) With this choice of $\epsilon$, finally set $\theta = \sqrt{\frac{1}{2} \lambda(\bm{x}_1,\ldots,\bm{x}_n)}$. We have 
\begin{align*}
    \| S(\bm{y}_1,\ldots,\bm{y}_n) \bm{\alpha} \|^2 \geq \lambda(\bm{y}_1, \ldots, \bm{y}_n) \bm{\alpha}^T \bm{\alpha} > \theta^2 \| \bm{\alpha} \|^2.
\end{align*}
\end{proof}

For proving the next theorem, we enforce the following additional condition on $\epsilon$ in the previous lemma. Note, if a particular $\epsilon > 0$ works in the above lemma, any smaller positive value will satisfy the lemma.

\begin{condition}
\label{cond_eps}
Let $\bm{x}_i, \bm{x}_j \in \Omega_n$ for $i, j = 1,\ldots,n$. The choice of $\epsilon$ in Lemma \ref{lem_knlMatBnd} also satisfies
\begin{align*}
    0 < \epsilon < \frac{1}{2} \min_{i\neq j} d_M(\bm{x}_i,\bm{x}_j).
\end{align*}
\end{condition}

\begin{lemma}
\label{lem_fixIntPE}
Let $I$ be a bounded, Lebesgue ($\mu$) measurable subset of $[0,\infty)$, and also let 
\begin{align*}
    I_i := \{ s \in I | d_M(\bm{x}_i,\bm{x}(t)) \leq \epsilon \text{ for } i = 1,\ldots,n \},
\end{align*}
where $\epsilon$ is as in Lemma \ref{lem_knlMatBnd} and satisfies Condition \ref{cond_eps}. If $\mu(I_i) \geq \tau_0$ for $1,\ldots,n,$ then with $\theta$ as in Lemma \ref{lem_knlMatBnd}, 
\begin{align*}
    \int_I \left( \mathcal{E}^*_{\bm{x}(\tau)} \mathcal{E}_{\bm{x}(\tau)}g,g \right)_{R_M(\mathcal{H}_X)} d\tau \geq \tau_0 \theta^2 \|\bm{\alpha}\|^2
\end{align*}
holds for any $g \in R_M(\mathcal{H}_X)$ and $\bm{\alpha} = \{ \alpha_1,\ldots,\alpha_n\}$ such that $g = \sum_{i=1}^n \alpha_i \mathcal{R}(\bm{x}_i,\cdot) $.
\end{lemma}

\begin{proof}
First, we note that
\begin{align*}
    \left\langle \mathcal{E}^*_{\bm{x}(\tau)} \mathcal{E}_{\bm{x}(\tau)}g,g \right \rangle_{R_M(\mathcal{H}_X)} = \left\langle \mathcal{E}_{\bm{x}(\tau)}g, \mathcal{E}_{\bm{x}(\tau)}g \right\rangle_\mathbb{R}
    = \left(g(\bm{x}(\tau))\right)^2. 
\end{align*}
Moreover, the sets $I_i$ are disjoint since the closed balls defined as $B_{\epsilon}(\bm{x}_i) := \{ \bm{y} \in M | d_M(\bm{x}_i,\bm{y}) \leq \epsilon \}$ centered at $x_i$ and radius $\epsilon$ do not intersect with each other when $\epsilon$ satisfies Condition \ref{cond_eps}. Furthermore, since $\cup_{i=1}^n I_i \subseteq I$, we have
\begin{align}
\label{eq_PE2lowerbnd}
    \int_I \left\langle \mathcal{E}^*_{\bm{x}(\tau)} \mathcal{E}_{\bm{x}(\tau)}g,g \right\rangle_{R_M(\mathcal{H}_X)} d\tau 
    \geq \sum_{i=1}^n \int_{I_i} \left( g(\bm{x}(\tau)) \right)^2 d\tau.
\end{align}
The closed balls $B_{\epsilon}(\bm{x}_i)$ are compact since the manifold $M$ is compact. Thus, the function $g \in R_M(\mathcal{H}_X) \subseteq C(M)$ attains its maximum and minimum at points in the manifold $M$, say $\overline{\bm{y}}_i, \underline{\bm{y}}_i \in M$, respectively. Thus, for each $i = 1,\ldots,n$, we get the inequality
\begin{align*}
    \left( g(\underline{\bm{y}}_i) \right)^2 \mu(I_i) \leq \int_{I_i} \left( g(\bm{x}(\tau)) \right)^2 d\tau \leq \left( g(\overline{\bm{y}}_i) \right)^2 \mu(I_i).
\end{align*}
By definition of $d_M$, we know that the closed ball $B_{\epsilon}(\bm{x}_i)$ is connected. Using the generalized intermediate value theorem and the hypothesis that $\mu(I_i) \geq \tau_0$, we conclude that there exists a $\bm{y}_i \in B_{\epsilon}(\bm{x}_i)$ such that
\begin{align*}
    \int_{I_i} \left( g(\bm{x}(\tau)) \right)^2 d\tau = \left( g(\bm{y}_i) \right)^2 \mu(I_i) \geq \left( g(\bm{y}_i) \right)^2 \tau_0
\end{align*}
for $i = 1,\ldots,n$. From the Inequality \ref{eq_PE2lowerbnd}, we have
\begin{align*}
    &\int_I \left\langle \mathcal{E}^*_{\bm{x}(\tau)} \mathcal{E}_{\bm{x}(\tau)}g,g \right\rangle_{R_M(\mathcal{H}_X)} d\tau 
    \\
    & \hspace{0.5in}
    \geq \sum_{i=1}^n \left( \sum_{j=1}^n \alpha_j \mathcal{R}(\bm{x}_j,\bm{y}_i) \right)^2 \tau_0
    = \| S \bm{\alpha} \|^2 \tau_0,
\end{align*}
where $S = S(\bm{y}_1,\ldots,\bm{y}_n)$ is defined as in Equation \ref{eq_knlMatDef}. Since we choose $\epsilon$ as in Lemma \ref{lem_knlMatBnd}, using the lemma gives us the desired result.
\end{proof}

The above lemma plays a direct role in the proof of the sufficient conditions for PE given below.

\begin{theorem}
\label{thm_suffPE}
Suppose that the manifold $M$ is positive invariant under the state trajectory $t \mapsto \bm{x}(t)$ and $\Omega_n := \{ \bm{x}_1,\ldots,\bm{x}_n \}$. Also suppose that the constant $\epsilon$ is chosen as in Lemma \ref{lem_knlMatBnd} and satisfies Condition \ref{cond_eps}. For every $t \geq 0$ and every $\Delta_2 > 0$, define 
\begin{align*}
    I_i := \{ s \in [t,t+\Delta_2] | d_M(\bm{x}_i,\bm{x}(t)) \leq \epsilon \text{ for } i = 1,\ldots,n \}.
\end{align*}
If there exists a $T_2 \geq 0$ and $\Delta_2 > 0$ such that for all $t \geq T_2$, $\mu(I_i)$ is bounded below by a positive constant $\tau_0 > 0$ for all $i = 1,\ldots,n$ and $t$, then the trajectory $\bm{x} : t \to \bm{x}(t)$ persistently excites the indexing set $\Omega_n$ and the RKHS $R_M(\mathcal{H}_{\Omega_n})$ in the sense of PE $R_M(\mathcal{H}_{\Omega_n})$ - \ref{def_PE2}.
\end{theorem}

\begin{proof}
For a given $t \geq T_2$, we define $I:= [t, t + \Delta_2 ]$. Since $\mu(I_i) \geq \tau_0$ for $i = 1,\ldots,n$, we apply Lemma \ref{lem_fixIntPE} to get
\begin{align*}
    \int_{t}^{t + \Delta_2} \left\langle \mathcal{E}^*_{\bm{x}(\tau)} \mathcal{E}_{\bm{x}(\tau)}g,g \right\rangle_{R_M(\mathcal{H}_X)} d\tau \geq \tau_0 \theta^2 \|\bm{\alpha}\|^2.
\end{align*}
We note that the constant $\tau_0$ is independent of $t$. Thus, the above inequality is valid for all $t \geq T_2$. Given $g = \sum_{i=1}^n \alpha_i \mathcal{R}(\bm{x}_i,\cdot) \in R_M(\mathcal{H}_X)$, its norm is given by
\begin{align*}
    \|g\|^2_{R_M(\mathcal{H}_X)} &= \left\langle g,g \right\rangle_{R_M(\mathcal{H}_X)} 
    \\
    &= \left\langle \sum_{i=1}^n \alpha_i \mathcal{R}(\bm{x}_i,\cdot),\sum_{i=1}^n \alpha_i \mathcal{R}(\bm{x}_i,\cdot) \right\rangle_{R_M(\mathcal{H}_X)} 
    \\
    &= \bm{\alpha}^T S(\bm{x}_1,\ldots,\bm{x}_n) \bm{\alpha},
\end{align*}
where $S(\bm{x}_1,\ldots,\bm{x}_n)$ is defined as in Equation \ref{eq_knlMatDef} and is called the Grammian matrix. It is straightforward to see that the norm in $S(\bm{x}_1,\ldots,\bm{x}_n)$ is equivalent to the norm in $\mathbb{R}^n$ since
\begin{align*}
    \underline{\lambda} \| \bm{\alpha} \|^2 \leq \|g\|^2_{R_M(\mathcal{H}_X)} \leq \overline{\lambda} \| \bm{\alpha} \|^2,
\end{align*}
where $\underline{\lambda}$ and $\overline{\lambda}$ are the minimum and maximum eigenvalues of the Grammian matrix $S(\bm{x}_1,\ldots,\bm{x}_n)$. Note that the Grammian matrix is a symmetric positive definite matrix, and hence all eigenvalues are real and positive. Using the above equivalence of norms, we get
\begin{align*}
    \int_{t}^{t + \Delta_2} \left\langle \mathcal{E}^*_{\bm{x}(\tau)} \mathcal{E}_{\bm{x}(\tau)}g,g \right\rangle_{R_M(\mathcal{H}_X)} d\tau \geq \gamma_2 \|g\|^2_{R_M(\mathcal{H}_X)},
\end{align*}
where $\gamma_2 = \frac{\tau_0 \theta^2}{\overline{\lambda}}$.
\end{proof}

Theorem \ref{thm_suffPE} states that after a finite amount of time $T_2$, if there exists a constant $\Delta_2$ such that in any time window $[t, t + \Delta_2] \subseteq [T_2,\infty)$, the state trajectory stays in the neighborhood of each of the centers $\bm{x}_1,\ldots,\bm{x}_n$ for at least a finite amount of time $\tau_0$, then the state trajectory is persistently exciting in the sense of PE defined in the theorem. The example in Section \ref{sec_numIllus} gives an intuitive illustration of the sufficient condition.

\begin{corollary}
\label{cor_suffPE2ParamConv}
If the hypothesis of Theorem \ref{thm_suffPE} holds and the family of functions $\mathbb{U}(\bar{S}_n)$, as defined in Subsection \ref{ssec_PEequiv}, are uniformly equicontinuous, then the trajectory $\bm{x} : t \to \bm{x}(t)$ persistently excites the indexing set $\Omega_n$ and the RKHS $R_M(\mathcal{H}_{\Omega_n})$ in the sense of PE $R_M(\mathcal{H}_{\Omega_n})$ - \ref{def_PE1}. 

Furthermore, if the unknown nonlinear function $f \in R_M(\mathcal{H}_{\Omega_n})$, then
\begin{align*}
    \lim_{t \to \infty} \| \tilde{\bm{x}}(t) \| = 0, \hspace{0.75in} \lim_{t \to \infty} \| \tilde{f}(t) \|_{R_M(\mathcal{H}_{\Omega_n})} = 0.
\end{align*}
\end{corollary}

The proof of the above corollary follows directly from Theorem \ref{thm_suffPE} and the discussion in Subsections \ref{ssec_PE} and \ref{ssec_PEequiv}.

%
\section{Implications of the Sufficient Condition in Infinite Dimensions}
\label{sec_suffCondImp}
In the previous section, we considered the case where the unknown nonlinear function $f$ is in the finite-dimensional space $R_M(\mathcal{H}_{\Omega_n})$. In this section, we consider the case where it is only known that $f \in R_M(\mathcal{H}_X)$. Since functions in $R_M(\mathcal{H}_X)$ are defined only on the manifold $M$, we need the state trajectory $\bm{x}(t)$ to be contained in $M$ for the governing equations to make sense. The implication of this change in hypotheses is that the function estimate error is ultimately bounded above by a constant which depends on the norm of the complementary projection.

In the following analysis, we use the subscript $n$ to denote the terms associated with the finite-dimensional space $R_M(\mathcal{H}_{\Omega_n})$. Let $\mathbb{P}_{\Omega_n}$ be the projection operator from $R_M(\mathcal{H}_X)$ onto $R_M(\mathcal{H}_{\Omega_n})$. The projection operator decomposes the space $R_M(\mathcal{H}_X)$ into $R_M(\mathcal{H}_X) = R_M(\mathcal{H}_{\Omega_n}) \bigoplus R_M(\mathcal{V}_{\Omega_n})$. Note that space $R_M(\mathcal{V}_{\Omega_n})$ contains functions that are orthogonal to the functions in $R_M(\mathcal{H}_{\Omega_n})$. Using the reproducing property, it is easy to show that functions in $R_M(\mathcal{V}_{\Omega_n})$ vanish identically on the set $\Omega_n$, i.e., for all $v \in R_M(\mathcal{V}_{\Omega_n})$, and $\bm{x} \in \Omega_n$, $v(\bm{x}) = 0$. With this definition, we rewrite the plant equation given in Equation \ref{eq_plant}, in which $f \in R_M(\mathcal{H}_X)$, in the form
\begin{align*}
\dot{\bm{x}}(t) = A \bm{x}(t) + B f_n(\bm{x}(t)) + B v_n(\bm{x}(t)),
\end{align*}
where $\bm{x}(t) \in M$, $f_n = \mathbb{P}_{\Omega_n} f \in R_M(\mathcal{H}_{\Omega_n})$, $v_n \in R_M(\mathcal{V}_{\Omega_n})$ and $f = f_n + v_n$. 

For practical applications, we want estimates that are finite-dimensional. We replace the infinite-dimensional estimate $\hat{f}$ in Equation \ref{eq_plntEst} with the finite-dimensional estimate $\hat{f}_n := \hat{f}_n(t,\cdot) \in R_M(\mathcal{H}_{\Omega_n})$. The estimator equation has the form
\begin{align}
    \dot{\hat{\bm{x}}}(t) = A \hat{\bm{x}}(t) + B \hat{f}_n(t,\bm{x}(t)),
    \label{eq_fderrState}
\end{align}
where $\hat{\bm{x}}(t) \in \mathbb{R}^d$ is the finite-dimensional state estimate. We use the dead-zone gradient learning law
\begin{align}
    \dot{\hat{f}}_n(t) = \Gamma^{-1} (B \mathcal{E}_{\bm{x}(t)} \mathbb{P}_{\Omega_n})^* \tilde{\bm{x}}_D(t),
    \label{eq_fderrFunc}
\end{align}
where
\begin{align*}
    \tilde{\bm{x}}_D(t) = \tilde{\bm{x}}(t) - \Phi \sigma(\tilde{\bm{x}}(t)).
\end{align*}
In the above equation, $\Phi = \frac{\|B\| \|v_n\|_{C(U)}}{\lambda_A}$,
and the other terms are defined as in Equation \ref{eq_lLaw}. The saturation function $\sigma: \mathbb{R}^d \to \mathbb{R}^d$ is defined as
\begin{align*}
    \sigma_i(\bm{x}) = 
    \left \{
    \begin{array}{cc}
        \frac{\bm{x}_i}{\Phi} & \text{ if } \left|\frac{\bm{x}_i}{\Phi}\right| \leq 1, \\
        1 & \text{ if } \left|\frac{\bm{x}_i}{\Phi}\right| > 1
    \end{array}
    \right.
\end{align*}
for $i = 1,\ldots,d$ and $\sigma(\bm{x}) = \{ \sigma_1(\bm{x}),\ldots,\sigma_d(\bm{x}) \}^T$.
Using the reproducing property, we can show that Equation \ref{eq_fderrFunc} is equivalent to 
\begin{align}
    \dot{\hat{\bm{\alpha}}}(t) &= S(\bm{x}_1,\ldots,\bm{x}_n)^{-1} \bm{\Gamma}^{-1} \bm{\mathcal{R}}(\bm{x}_{c},\bm{x}(t)) B^T \tilde{\bm{x}}_D(t),
    \label{eq_paramLLaw}
\end{align}
where $S(\bm{x}_1,\ldots,\bm{x}_n)$ is defined as in Equation \ref{eq_knlMatDef}, $\bm{\mathcal{R}}(\bm{x}_{c},\bm{x}(t)) := \{ \bm{\mathcal{R}}(\bm{x}_1,\bm{x}(t)), \ldots, \bm{\mathcal{R}}(\bm{x}_n,\bm{x}(t)) \}^T$, and $\bm{\Gamma}:= \Gamma \mathbb{I}_{n}$ is the gain matrix. \cite{Paruchuri2020} The term $\hat{\bm{\alpha}}(t) := \{ \alpha_1(t),\ldots,\alpha_n(t) \}^T$ in the above equation is the unknown parameter that satisfies $\hat{f}_n = \sum_{i=1}^n \alpha_i \mathcal{R}(\bm{x}_i,\cdot)$. The error equations are then
\begin{align}
    \begin{Bmatrix}
        \dot{\tilde{\bm{x}}}(t) \\
        \dot{\tilde{f}}(t)
    \end{Bmatrix}
    &=
    \begin{Bmatrix}
        A \tilde{\bm{x}}(t) + B \mathcal{E}_{\bm{x}(t)} \tilde{f}(t)
        \\
        - \Gamma^{-1} (B \mathcal{E}_{\bm{x}(t)} \mathbb{P}_{\Omega_n})^* \tilde{\bm{x}}_D(t)
    \end{Bmatrix}
    \label{eq_errEst_FD}
\end{align}
where $\tilde{\bm{x}}(t) := \bm{x}(t) - \hat{\bm{x}}(t)$ and $\tilde{f}(t):= f - \hat{f}_n(t)$. If we define $\tilde{f}_n(t):= \mathbb{P}_{\Omega_n} f - \hat{f}_n = f_n - \hat{f}_n(t)$, we get $\tilde{f}(t):= \tilde{f}_n(t) + v_n$. Since the function $v_n$ is a constant, we have $\dot{\tilde{f}}(t) = \dot{\tilde{f}}_n(t)$.

The following theorem shows us that the sufficient condition given in Theorem \ref{thm_suffPE} implies boundedness of the error by a constant proportional to $\|(I - \mathbb{P}_{\Omega_n})f\|_{C(U)}$, where $U$ is a compact set in which the states are contained after a finite amount of time. The proof of the theorem is similar to that of Theorem \ref{thm_paramConv}. In the context of this theorem, we assume that the state trajectory $t \mapsto \bm{x}(t)$ is \emph{bounded} and \emph{uniformly continuous}. Generally speaking, we use both these assumptions in the proof of Theorem \ref{thm_suffPE}.

\begin{theorem}
\label{thm_suffPEImp}
Suppose that the state trajectory $\bm{x}(t) \in M$, the function $f \in R_M(\mathcal{H}_X)$, and the class of functions $\mathbb{U}(\bar{S}_n)$ defined in Subsection \ref{ssec_PEequiv} is uniformly equicontinuous. Also suppose that the constant $\epsilon$ is chosen as in Lemma \ref{lem_knlMatBnd} and satisfies Condition \ref{cond_eps}. 
If 
the sufficient condition given by Theorem \ref{thm_suffPE} holds,
and the evolution of $\hat{f}_n(t)$ is governed by Equation \ref{eq_fderrFunc},
then 
\begin{align*}
    \limsup \limits_{t\rightarrow \infty} \|\tilde{\bm{x}}(t)\| &\leq \hat{c} \| v_n \|_{C(U)}, 
    \\
    \limsup \limits_{t \to \infty} \| \tilde{f}_n(t) \| &\leq \check{c} \| v_n \|_{C(U)},
\end{align*}
where $\hat{c}:=\hat{c}(n)$ and $\check{c}:=\check{c}(n)$ are constants and $\| v_n \|_{C(U)}$ denotes the uniform norm of the function $v_n$ over the set $U = \overline{\{ \bm{x}(\tau) | \tau \geq T_1  \}}$ with $T_1$ is defined as in PE $R_M(\mathcal{H}_{\Omega_n})$-\ref{def_PE1}.
\end{theorem}

\begin{proof}
Since the hypotheses of Corollary \ref{cor_suffPE2ParamConv} holds, the state trajectory is persistently exciting in the sense that there exist constants $T_1, \gamma_1, \delta_1$ and $\Delta_1$ such that the PE condition given in the corollary holds. Consider the Lyapunov function
\begin{align*}
    V(t)=\left \langle\tilde{\bm{x}}_D(t),\tilde{\bm{x}}_D(t) \right \rangle +  \left \langle \tilde{f}_n(t), \Gamma \tilde{f}_n(t) \right \rangle_{R_M(\mathcal{H}_X)}.
\end{align*}
The time derivative of the Lyapunov equation is 
\begin{align*}
    \dot{V}(t) &= \left\langle \tilde{\bm{x}}_D(t),\left( I - \Phi \frac{\partial \sigma}{\partial \bm{x}}(\tilde{\bm{x}}(t)) \right) \dot{\tilde{\bm{x}}}(t)\right\rangle + \left\langle \tilde{f}_n(t), \dot{\tilde{f}}_n(t) \right\rangle \\
    &= -\lambda_A \| \tilde{\bm{x}}_D(t) \|^2 - \lambda_A \Phi \| \tilde{\bm{x}}_D(t) \|_1 + \tilde{\bm{x}}_D(t)^T B v_n(\bm{x}(t)).
\end{align*}
In deriving the above equation, we assume that the matrix $A$ has the form $A = -\lambda_A I$, where $\lambda_A > 0$. (There is not loss of generality in this assumption. If $A$ does not satisfy this assumption, we can modify the estimator in Equation \ref{eq_fderrState} by replacing the term $A \hat{\bm{x}}(t)$ with $A \bm{x}(t)$ and adding the term $\lambda_A I \tilde{\bm{x}}(t)$. Then the error equations have the same form as in Equation \ref{eq_errEst_FD} and the analysis proceeds without change.)
Since $\Phi = \frac{\|B\| \|v_n\|_{C(U)}}{\lambda_A}$, we conclude that
\begin{align*}
    \dot{V}(t) \leq -\lambda_A \| \tilde{\bm{x}}_D(t) \|^2.
\end{align*}
Thus, we conclude that $\tilde{\bm{x}}_D(t)$, $\tilde{\bm{x}}(t)$ are bounded and the family of functions $\{ \tilde{f}_n(t) \}_{t \geq 0}$ is uniformly bounded.

Next note that $\dot{\tilde{\bm{x}}}(t)$ is bounded. This is evident from the equality
\begin{align*}
    \|\dot{\tilde{\bm{x}}}\| \leq \|A\|\|\tilde{\bm{x}}(t)\| + \|B\| \|\mathcal{E}_{\bm{x}(t)}\| \left( \|\tilde{f}_n(t)\| + \|v_n\| \right).
\end{align*}
Thus, $\tilde{\bm{x}}(t)$ is Lipschitz continuous in $t$, which implies that the same is uniformly continuous in $t$. This in turn implies that $\tilde{\bm{x}}_D(t)$ is uniformly continuous in $t$. Next, notice that the function $v_n$ is bounded and uniformly continuous on the set $U$, since $v_n$ is continuous and $U$ is compact. Furthermore, recall that the state trajectory is bounded and uniformly continuous in $t$. Thus, $\dot{V}(t)$ is uniformly continuous in $t$. 

Since $V(t)$ is monotonically decreasing and bounded below, we have
\begin{equation*}
    \lim_{t \to \infty} \int_{t_0}^{t} \dot{V}(t) d\tau = \lim_{t \to \infty} V(t) - V(t_0) < \infty.
\end{equation*}
Using Barbalat's lemma for $\dot{V}$, we get 
\begin{align*}
    \lim_{t \to \infty} \|\tilde{\bm{x}}_D(t)\| = 0,
\end{align*}
which implies
\begin{align*}
    \limsup \limits_{t \to \infty} \| \tilde{\bm{x}}(t) \| \leq \hat{c} \| v_n \|_{C(U)},
\end{align*}
where $\hat{c} = \frac{\| B \| \sqrt{\bar{\lambda} n} }{\lambda_A}$, $\bar{\lambda}$ is the largest eigenvalue of the Grammian matrix $S(\bm{x}_1,\ldots,\bm{x}_n)$.

Next, we turn to the proof that $\limsup \limits_{t \to \infty} \| \tilde{f}_n(t) \| \leq \check{c} \| v_n \|_{C(U)}$. Given $\varepsilon > 0$, there exists a $T$ such that for all $t \geq T$, $\| \tilde{\bm{x}}_D(t) \| < \varepsilon$ or $\| \tilde{\bm{x}}(t) \| < \hat{c} \| v_n \|_{C(U)} + \varepsilon$. Without loss of generality, select the constant $T \geq T_1$. Let $s \in [T, T+\Delta_1]$. Since we know how the state error evolves, the norm of the state error at $s+\delta_1$ is bounded below by
{
\begin{align*}
    &
    \|\tilde{\bm{x}}(s+\delta_1)\| = \left\|\tilde{\bm{x}}(s) + \int_{s}^{s+\delta_1} A\tilde{\bm{x}}(\tau) + B\mathcal{E}_{\bm{x}(\tau)} \tilde{f}(\tau) d\tau \right\|, \nonumber 
    \\
    &\hspace{0.5em}
    \geq \underbrace{\left\|\int_{s}^{s+\delta_1} B\mathcal{E}_{\bm{x}(\tau)} \tilde{f}(T) d\tau\right\|}_{\text{term 1}} - \underbrace{\left\|\tilde{\bm{x}}(s) + \int_{s}^{s+\delta_1} A\tilde{\bm{x}}(\tau) d\tau\right\|}_{\text{term 2}} \nonumber 
    \\ 
    &\hspace{4em}
    - \underbrace{\left\|\int_{s}^{s+\delta_1} B\mathcal{E}_{\bm{x}(\tau)} (\tilde{f}(\tau)-\tilde{f}(T)) d\tau\right\|}_{\text{term 3}}.
\end{align*}
}
Let us consider term 1. We have
\begin{align*}
    &\left\|\int_{s}^{s+\delta_1} B\mathcal{E}_{\bm{x}(\tau)} \tilde{f}(T) d\tau\right\| = \|B\| \left|\int_{s}^{s+\delta_1} \mathcal{E}_{\bm{x}(\tau)} \tilde{f}(T) d\tau\right|,
    \\
    &= \|B\| \left|\int_{s}^{s+\delta_1} \mathcal{E}_{\bm{x}(\tau)} \left( \tilde{f}_n(T) + v_n \right) d\tau\right|,
    \\
    &\geq \|B\| \left|\int_{s}^{s+\delta_1} \mathcal{E}_{\bm{x}(\tau)} \tilde{f}_n(T) d\tau\right| - \|B\| \left|\int_{s}^{s+\delta_1} v_n(\bm{x}(\tau)) d\tau\right|.
\end{align*}
Since the PE condition for $R_M(\mathcal{H}_{\Omega_n})$ is valid, we have
{\small
\begin{align*}
    \left\|\int_{s}^{s+\delta_1} B\mathcal{E}_{\bm{x}(\tau)} \tilde{f}(T) d\tau\right\|
    &\geq \gamma_1 \|B\| \| \tilde{f}_n(T) \| - \delta_1 \|B\| \| v_n \|_{C(U)}.
\end{align*}
}%
Now consider term 2. We bound term 2 above by
\begin{align*}
    &\left\|\tilde{\bm{x}}(s) + \int_{s}^{s+\delta_1} A\tilde{\bm{x}}(\tau) d\tau\right\| 
    \\
    & \hspace{2em}
    \leq \|\tilde{\bm{x}}(s) \| + \int_{s}^{s+\delta_1} \|A\| \| \tilde{\bm{x}}(\tau) \| d \tau,
    \\
    & \hspace{2em} 
    \leq (1 + \|A\| \delta) (\hat{c} \| v_n \|_{C(U)} + \varepsilon).
\end{align*}%
Before proceeding further, we consider the term $\tilde{f}(\tau)-\tilde{f}(T)$. Using the learning law gives us
\begin{align}
    &\|\tilde{f}(\tau)-\tilde{f}(T)\|_{R_M(\mathcal{H}_X)} 
    \notag
    \\
    & \hspace{2em}
    = \left\| \int_T^\tau \Gamma^{-1} (B \mathcal{E}_{\bm{x}(\xi)} \mathbb{P}_{\Omega_n})^* \tilde{\bm{x}}_D(\xi) d\xi \right\|_{R_M(\mathcal{H}_X)},
    \label{eq_thmSuffPEImpTerm2}
    \\
    & \hspace{2em}
    \leq \int_T^\tau \Gamma^{-1} \|B\| \|\mathcal{E}_{\bm{x}(\xi)}\| \| \mathbb{P}_{\Omega_n} \| \|\tilde{\bm{x}}_D(\xi)\| d \xi,
    \notag
    \\
    & \hspace{2em}
    \leq c_2 (\tau - T) \varepsilon,
    \notag
\end{align}
where $c_2 = \Gamma^{-1} \|B\| \|\mathcal{E}_{\bm{x}(\xi)}\|$. Thus, since $T \leq s \leq T+\Delta$, Term 3 is bounded above by
\begin{align*}
    &
    \left\|\int_{s}^{s+\delta_1} B\mathcal{E}_{\bm{x}(\tau)} (\tilde{f}(\tau)-\tilde{f}(T)) d\tau\right\|
    \\
    & \hspace{2em}
    \leq \int_s^{s+\delta} \| B \| \| \mathcal{E}_{\bm{x}(\tau)} \| \|\tilde{f}(\tau)-\tilde{f}(T)\|_{R_M(\mathcal{H}_X)} d\tau
    \leq c_3 \varepsilon,
\end{align*}
where $c_3 = \| B \| \| \mathcal{E}_{\bm{x}(\tau)} \| c_2 \left( \frac{1}{2} \delta_1^2 + \Delta_1 \delta_1 \right)$. Thus, the norm of the state error at $s+\delta_1$ is bounded below by
\begin{align*}
    &\|\tilde{\bm{x}}(s+\delta_1)\| \geq \gamma_1 \|B\| \| \tilde{f}_n(T) \| - \delta_1 \|B\| \| v_n \|_{C(U)} 
    \\
    & \hspace{0.5in}
    - (1 + \|A\| \delta_1) (\hat{c} \| v_n \|_{C(U)} + \varepsilon) - c_3 \varepsilon.
\end{align*}
Since $s+\delta_1 > s \geq T$, we know that $\|\tilde{\bm{x}}(s+\delta_1)\|< \hat{c} \| v_n \|_{C(U)} + \varepsilon$. Rearranging the terms in the above equation, we get
\begin{align*}
    \| \tilde{f}_n(T) \| < \check{c} \| v_n \|_{C(U)} + \frac{(2 + \|A\| \delta_1 + c_3)}{\gamma_1 \|B\|} \varepsilon,
\end{align*}
where $\check{c} = \frac{(2 + \|A\| \delta_1) \hat{c} + \delta_1 \|B\| }{\gamma_1 \|B\| }$. In the argument above, $T$ is such that $\| \tilde{x}_n(t) \| < \hat{c} \| v_n \|_{C(U)} + \varepsilon$ for all $t \geq T$. We can repeat the above analysis for a sequence of $\varepsilon$, $\{ \bar{\varepsilon}_k \}_{k=1}^\infty$ such that $\bar{\varepsilon}_1 > \bar{\varepsilon}_2 > \ldots$, $\lim_{k \to \infty} \bar{\varepsilon}_k \to 0$. We can find an associated sequence of $T$, $\{ \bar{T}_k \}_{k=1}^\infty$ such that $\bar{T}_1 < \bar{T}_2 < \ldots$, $\lim_{k \to \infty} \bar{T}_k \to \infty$. Note that for any $\tau$ such that $\bar{T}_k \leq \tau < \bar{T}_{k+1}$, we have
\begin{align*}
    \| \tilde{f}_n(\tau) \| < \check{c} \| v_n \|_{C(U)} + \frac{(2 + \|A\| \delta_1 + c_3)}{\gamma_1 \|B\|} \bar{\varepsilon}_k.
\end{align*}
Thus, we conclude that
\begin{align*}
    \limsup \limits_{t \to \infty} \| \tilde{f}_n(t) \| \leq \check{c} \| v_n \|_{C(U)}. 
\end{align*}
\end{proof}

Remarks on Theorem \ref{thm_suffPEImp}:
\begin{enumerate}
    \item The implication of Theorem \ref{thm_suffPEImp} agrees with our intuition. The term $\| \tilde{f}_n(t) \|$ is eventually bounded by a constant that depends on the norm of the orthogonal component $v_n = (I - \mathbb{P}_{\Omega_n}) f$ of the unknown function, the matrix $A$ and the dimension $n$. In particular, the bound depends on the uniform norm of $v_n$ on the set $U$.
    
    \item By comparing Theorems \ref{thm_paramConv} and \ref{thm_suffPEImp}, it is clear that both theorems rely on different notions of PE (PE $R_M(\mathcal{H}_X)$-2 and PE $R_M(\mathcal{H}_{\Omega_n})$-2, respectively), which leads to distinct results. The differing notions arise from the fact that the we assume that the actual function $f$ is an element of $R_M(\mathcal{H}_{\Omega_n})$, as opposed to $R_M(\mathcal{H}_X)$, while proving the sufficient condition in Theorem \ref{thm_suffPE}.
    
    \item In the above theorem, we can interpret $v_n$ as process noise. 
    
\end{enumerate}

We next study in the following corollary a new error bound that is an immediate result of the sufficient condition derived in Theorem \ref{thm_suffPE} and Theorem \ref{thm_suffPEImp}. Following the derivation of the corollary, we compare and contrast the nature of the new convergence rate with the results in \cite{Guo2020Rates}.

\begin{corollary}
\label{cor_suffPEImpLip}
Suppose that the hypothesis of Theorem \ref{thm_suffPEImp} holds, and suppose that the set $U \subseteq \cup_{i=1}^n B_\eta (\bm{x}_i)$, where $B_\eta (\bm{x}_i)$ is the closed ball of radius $\eta$ centered at $\bm{x}_i \in \Omega_n$. If the function $v_n$ is Lipschitz continuous on $\cup_{i=1}^n B_\eta (\bm{x}_i)$, then 
\begin{align*}
    \lim_{t \to \infty} \| \tilde{f}_n(t) \| \leq \check{c} L \eta,
\end{align*}
where $L$ is the Lipschitz constant and $\check{c}:=\check{c}(n)$ is a constant defined as in Theorem \ref{thm_suffPEImp}.
\end{corollary}

\begin{proof}
We know that the function $v_n$ vanishes identically on the set $\Omega_n$ since $v_n \in R_M(\mathcal{V}_{\Omega_n})$, i.e., $v_n(\bm{x}_i) = 0$ for all $\bm{x}_i \in \Omega_n$. Since  $v_n$ is Lipschitz continuous, we have
\begin{align*}
    |v_n(\bm{y})| \leq L d_M(\bm{x}_i,\bm{y}) \leq L \eta,
\end{align*}
where $\bm{x}_i \in \Omega_n$, $\bm{y} \in B_\eta (\bm{x}_i)$ for all $i = 1,\ldots,n$. Since the upper bound $L\eta$ is independent of index $i$ and the hypotheses of Theorem \ref{thm_suffPEImp} hold, we have 
\begin{align*}
    \lim_{t \to \infty} \| \tilde{f}_n(t) \| \leq \check{c} L \eta.
\end{align*}
\end{proof}

Remarks on Corollary \ref{cor_suffPEImpLip}:
\begin{enumerate}
    \item The corollary shows that if the function $v_n$ is Lipschitz continuous and if $U \subseteq \cup_{i=1}^n B_\eta (\bm{x}_i)$, then the error bound for $\lim_{t \to \infty} \| \tilde{f}_n(t) \|$ depends on the radius $\eta$ of the closed balls $B_\eta(\bm{x}_i)$. It is clear that the radius $\eta$ depends on the maximum distance of the state $\bm{x}(t)$ from the kernel centers $\bm{x}_i \in \Omega_n$ for $t \geq T_1$. Thus, by choosing $T_1$ large enough, we can make the error bound small. However, we cannot make the error bound zero since the radius $\eta$ depends on the distance between the neighboring kernel centers. If the distance between neighboring kernels centers is greater than $2\eta$, then the exists a $t$ such that $\bm{x}(t) \notin U \subseteq \cup_{i=1}^n B_\eta (\bm{x}_i)$. This suggests that we can try to reduce the error bound by choosing more kernel centers that are persistently excited. However, careful study of the constant $\check{c}$ shows that it depends on the number of kernel centers $n$, $\check{c}:=\check{c}_n$. A rigorous treatment of this strategy will require the control of the product $\check{c}_n \eta_n$. 
    
    \item The notion of distance between kernel centers directly ties to the concept of fill distance $h_{\Omega_n,M}$ defined as
    \begin{align*}
        h_{\Omega_n,M} := \sup_{\bm{y}\in M}\min_{\bm{x}_i\in\Omega_n} d_M(\bm{x}_i,\bm{y}).
    \end{align*}
    It is shown in \cite{Guo2020Rates} that for certain kernels, the rate of convergence of the finite-dimensional function estimate $\hat{f}_n$ to the infinite-dimensional function estimate $\hat{f}$ depends on the fill distance. We can think of the infinite-dimensional estimate as the one that makes the error $P_\Omega \tilde{f}(t) \to 0$ as $t \to \infty$, where $\Omega$ is the PE set (that consists of infinite number of elements). By adding more centers, the finite-dimensional $\hat{f}_n(t) = P_{\Omega_n} \hat{f}(t)$ converges to the infinite-dimensional function estimate $P_\Omega \tilde{f}(t)$. This in turn implies that the bound on the error goes to zero as we add more centers.

    \item In Corollary \ref{cor_suffPEImpLip}, we assumes that the function $v_n$ is Lipschitz continuous on the set $U$. This condition is equivalent to assuming that the change of the function $v_n$ is constrained in the set $U$. We can come up with conditions that ensure Lipschitz continuity in a variety of ways. Suppose that the kernel generates an RKHS $\mathcal{H}_X$ that is embedded in a Sobolev space $W^{s,2}(X)$. A well known example of such a kernel is the Sobolev-Matern kernel. If the Sobolev space is of high enough order, Sobolev embedding theorem implies that the space $W^{s,2}(X)$ is embedded in the Holder space $C^{(1,0)}(X)$. We know that functions in $C^{(1,0)}(X)$ are globally Lipschitz. This implies that the function $v_n$ is Lipschitz continuous.
\end{enumerate}

\section{Numerical Example}
\label{sec_numIllus}

To interpret and evaluate the implications of the sufficient condition, we consider the undamped, unforced version of the nonlinear piezoelectric oscillator studied in \cite{Paruchuri2020}. We show that the sufficient condition implies ultimate boundedness of the function error estimate when we implement a gradient learning law based adaptive estimator. The governing equations of this oscillator have the form

\begin{align}
\begin{split}
    \begin{Bmatrix}
    \dot{x}_1(t) \\ \dot{x}_2(t)
    \end{Bmatrix}
    &=
    \underbrace{\begin{bmatrix}
    0 & 1 \\
    -\frac{\hat{K}}{M} & 0
    \end{bmatrix}}_{A}
    \underbrace{
    \begin{Bmatrix}
    x_1(t) \\ x_2(t)
    \end{Bmatrix}
    }_{\bm{x}(t)}
    \\
    & \hspace{0.5in}
    +
    \underbrace{\begin{Bmatrix}
    0 \\ 1
    \end{Bmatrix}}_{B}
    \underbrace{ 
    \left(
    - \frac{\hat{K}_{N_1}}{M} x_1^3(t) - \frac{\hat{K}_{N_2}}{M} x_1^5(t) 
    \right)
    }_{f(\bm{x}(t))},
\end{split}
\label{eq_piezomodel}
\end{align}
where $M, \hat{K} $ are the modal mass and modal stiffness of the piezoelectric oscillator, respectively. The variables $\hat{K}_{N_1}, \hat{K}_{N_2}$ are constants derived from nonlinear piezoelectric constitutive laws. \cite{Paruchuri2020} The states $x_1(\cdot)$ and $x_2(\cdot)$ are the modal displacement and modal velocity, respectively. Typically, the two states are not of the same order of magnitude, which inspires the use of anisotropic kernel functions, i.e. those that are elongated in one direction. However, equivalently, it is much easier to introduce a scaling factor for the one of the states. We substitute $x_1(t) = s \tilde{x}_1(t)$ in the governing equations, where $s$ is the scaling factor, and redefine $\bm{x}(t) := [\tilde{x}_1(t), x_2(t)]^T$. For our simulations, we choose $M = 0.9745$, $\hat{K} = 329.9006$, $\hat{K}_{N_1} = -1.2901e+05$, $\hat{K}_{N_2} = 1.2053e+09$ and $s = 0.02$. Furthermore, we choose the initial condition $\bm{x}_0 = [\tilde{x}_1(0), x_2(0)]^T = [0.05,0]^T$.

Figure \ref{fig_suffCondStates} shows the evolution of the states with time. It is clear that the positive limit set for the selected initial condition is a smooth, compact, Riemmanian, $1$-dimensional manifold embedded in $X = \mathbb{R}^2$. In our simulations, we use the RKHS generated by the Sobolev-Matern $3,2$ kernel, which has the form
{\small
\begin{align}
    \mathcal{R}_{3,2}(\bm{x},\bm{y}) &= \left(1 + \frac{\sqrt{3} \|\bm{x}-\bm{y}\|}{l}\right)\exp{\left(-\frac{\sqrt{3} \|\bm{x}-\bm{y}\|}{l}\right)},
    \label{eq_SobMatKern}
\end{align}
}%
where $l$ is the scaling factor of length. \cite{Rasmussen2003} 

The adaptive estimation equations are given by Equation \ref{eq_fderrState}, and
\begin{align}
    \dot{\hat{f}}_n(t) = \Gamma^{-1} (B \mathcal{E}_{\bm{x}(t)})^* P (\bm{x}(t) - \hat{\bm{x}}(t)).
\end{align}
Notice that the above equation specifies the derivative of the function estimate. Using the reproducing property, we can show that this evolution law is equivalent to
\begin{align}
    \dot{\hat{\bm{\alpha}}}(t) &= S(\bm{x}_1,\ldots,\bm{x}_n)^{-1} \bm{\Gamma}^{-1} \bm{\mathcal{R}}(\bm{x}_{c},\bm{x}(t)) B^T P \tilde{\bm{x}}(t),
    \label{eq_paramLLaw2}
\end{align}
where all the terms are defined as in Equation \ref{eq_paramLLaw}.
To build the adaptive estimate, we fix $n$, then choose kernel centers $\bm{x}_1,\ldots,\bm{x}_n$ along with the gain parameter $\Gamma$, and integrate Equations \ref{eq_fderrState} and \ref{eq_paramLLaw2}. 

Figure \ref{fig_suffCondStates} depicts the state evolution with time as well as the positive limit set of our example. It is clear from the figure that the state trajectory is uniformly continuous. Our goal is to choose $n$ kernel centers $\bm{x}_1,\ldots,\bm{x}_n$ that are persistently excited. 
First let us note that the trajectory is periodic. Set $\Delta_2 = 2 t_p$, where $t_p$ is the period of the state trajectory. Consider an arbitrary point $\bm{x}_1$ in the positive limit set. Consider the window $I_p = [t,t+ 2t_p]$ for any arbitrary $t \geq 0$. It is clear that the time spent by the state trajectory in $B_\epsilon(\bm{x}_1)$ during any window $I_p$ is bounded below by a constant. In Figure \ref{fig_suffCondStates}, consider the (cyan) ball in the phase plane and any part of the state trajectory that is contained in a time window of $2 t_p$. It is clear that the time spent by the trajectory in this ball is bounded below. Thus, using Corollary \ref{cor_suffPE2ParamConv}, we conclude that the point $\bm{x}_1$ is persistently excited. We repeat this analysis until $n$ points are determined. For this specific problem, any discrete finite number of points in the positive limit set are persistently excited. Note that in our previous analysis, we did not explicitly calculate the radius $\epsilon$. However, the above analysis is valid for a ball of any positive radius centered at a point in the positive limit set. For a point outside the positive limit set, we need explicit knowledge of $\epsilon$ that is as in Lemma \ref{lem_knlMatBnd} and satisfies Condition \ref{cond_eps}. 

In the above analysis, we treat the state trajectory as elements contained in $\mathbb{R}^2$. However, the state trajectory is contained in the positive limit set, which is a smooth, compact, Riemmanian $1$-dimensional manifold $M$ embedded in $X = \mathbb{R}^2$. We can treat the problem as evolution on a manifold and restrict the Hilbert space of function $\mathcal{H}_X$ to the manifold. Analysis similar to the one given above holds in this case. The primary difference is that we consider closed balls that are contained in the one-dimensional manifold $M$ as opposed to ones contained in $\mathbb{R}^2$. We can determine the persistently excited points in $M$ and combine our analysis given in \cite{Guo2020Rates} to determine approximation rates of convergence.

Figure \ref{fig_Ex1Unif} depicts the pointwise error $|f(\bm{x}) - \hat{f}_n(t_e, \bm{x})|$ after running the adaptive estimator for $t_e = 150$ seconds with $50$ kernel centers initialized at $\alpha_i(0) = 0.001$ for all $i = 1,\ldots,50$. In our simulations, we set $\Gamma = 0.001$ and $l = 0.005$. Note that the function $f(\bm{x})$ in Equation \ref{eq_piezomodel} is clearly not in the space of $\mathcal{H}_{\Omega_{50}}$, where $\Omega_{50}$ is the set of $50$ kernel centers in the positive limit set denoted by the marker $*$ in Figure \ref{fig_Ex1Unif}. No linear combination of kernels, given by Equation \ref{eq_SobMatKern}, centered at points in $\Omega_{50}$ will be equation to $f(\bm{x})$. Thus, based on our analysis in Section \ref{sec_suffCondImp}, we can only guarantee boundedness of the asymptotic function error in the neighborhood of the positive limit set. Figure \ref{fig_Ex1Unif} clearly shows that the pointwise error is bounded around the positive limit set. Note, in our theorems imply convergence in the $R_M(H_X)$ norm. However, in an RKHS, convergence in RKHS norm implies pointwise convergence. In fact, since we consider only RKHS that are uniformly bounded, convergence in RKHS norm implies uniform convergence.

\begin{figure}
\centering
\includegraphics[scale = 0.65]{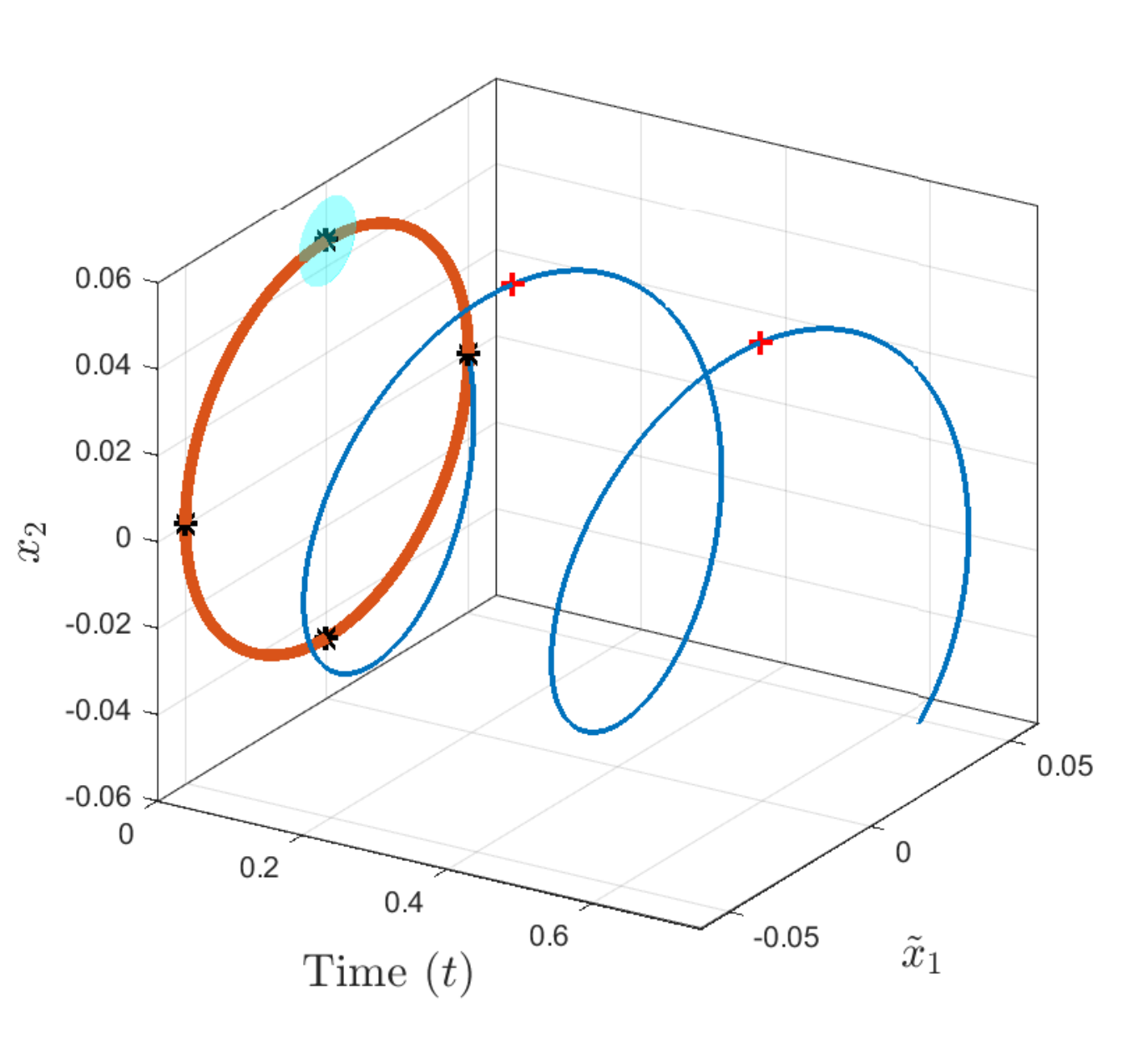}
\caption{State trajectory of the nonlinear system governed by Equation \ref{eq_piezomodel}, when the initial condition is $\bm{x}_0 = [0.05,0]^T$. The red loop is the positive limit set. The cyan circle represents the closed ball centered at the point depicted by marker $*$ in the phase plane. Marker $+$ represents the point depicted by $*$ in the state trajectory.}
\label{fig_suffCondStates}
\end{figure}

\begin{figure}
\centering
\includegraphics[scale = 0.65]{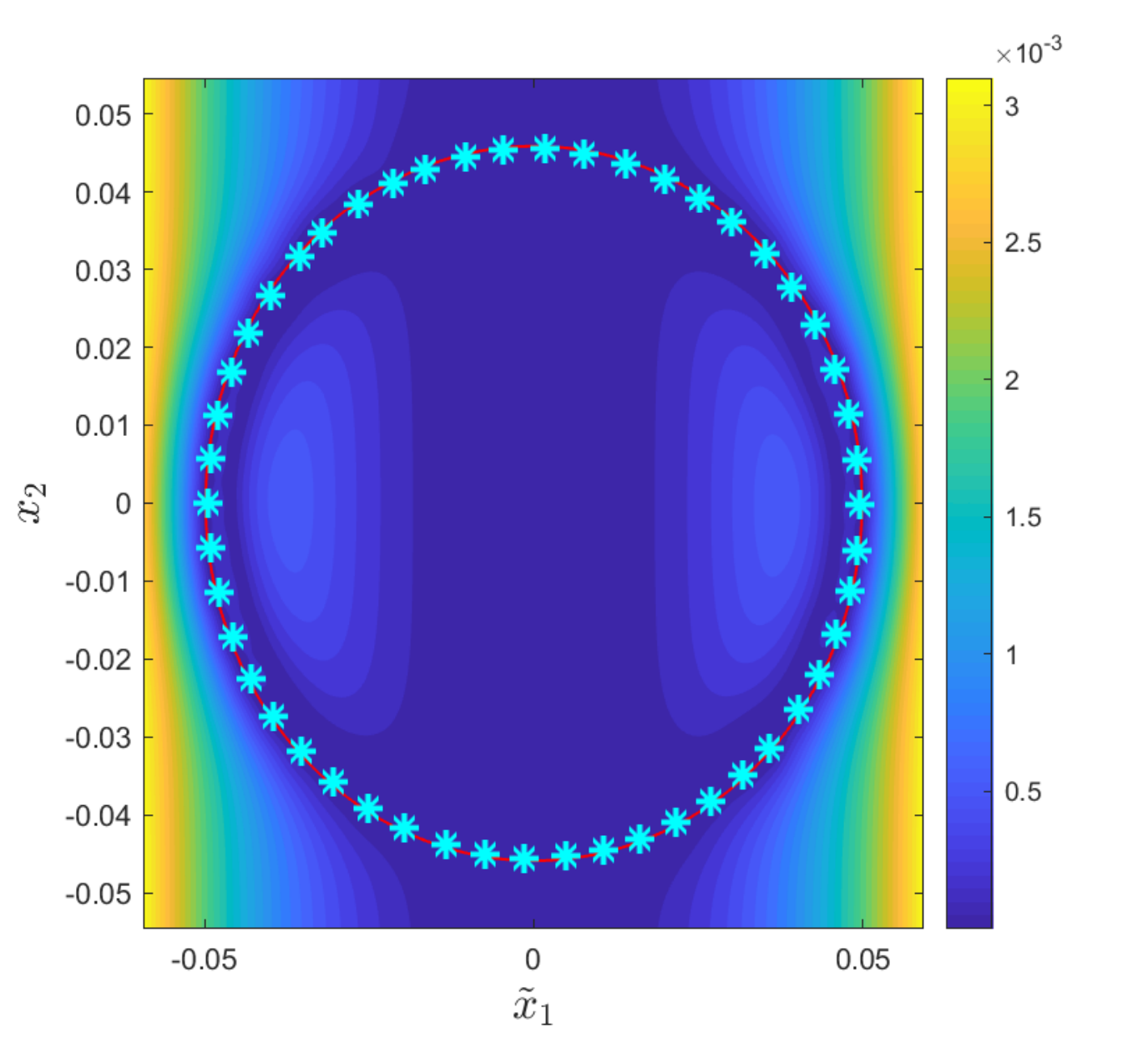}
\caption{Pointwise error $|f(\bm{x}) - \hat{f}_n(t_e, \bm{x})|$. The marker $*$ and the red line represent the kernel centers and the limit set $\Omega$, respectively.}
\label{fig_Ex1Unif}
\end{figure}

\section{Conclusion}
\label{sec_conc}
In this paper, we have derived a sufficient condition for different notions of PE in RKHS defined over embedded manifolds. This sufficient condition is valid for RKHS generated by continuous, strictly positive definite kernels. We have studied the implications of the sufficient condition in the case when the RKHS is finite or infinite-dimensional. When the unknown function resides in a finite-dimensional RKHS, the sufficient condition implies convergence of function error estimate. In the more general case when we only know that the unknown function resides in an infinite-dimensional RKHS, the sufficient conditions implies ultimate boundedness of the function estimate error by a constant that depends on the approximation error. Finally, the numerical example has illustrated the practicality of the sufficient condition.

\bibliographystyle{IEEEtran}
\bibliography{j_PE_suff_2020}

%








\end{document}